%% file: fscd.tex
\documentclass[a4paper,UKenglish,cleveref, autoref, thm-restate]{lipics-v2021}
\usepackage{tikz-cd}
\usetikzlibrary{trees}
\usepackage{mathrsfs}

\DeclareMathOperator{\rd}{\mathrm{rd}}
\DeclareMathOperator{\sk}{\mathrm{sk}}
\DeclareMathOperator{\nf}{\mathrm{nf}}
\DeclareMathOperator{\BT}{\mathrm{BT}}
\def\th{\EuScript}
\def\<{\langle}
\def\>{\rangle}
\newcommand{\norm}[1]{\nu(#1)}
\newcommand{\lam}{\ensuremath{\lambda}}
\newcommand{\Hst}{\EuScript{H}^*}
\newcommand\FHP{\mathrm{FHP}}
\newcommand\HP{\mathrm{HP}}

\newcommand\PT{\mathrm{PT}}
\newcommand\FPT{\mathrm{FPT}}
\newcommand\fin{\mathrm{fin}}

\title{Groups and Inverse Semigroups in Lambda Calculus} 

\author{Antonio Bucciarelli}{Université Paris Cité, CNRS, IRIF, F-75013, Paris, France}{buccia@irif.fr}{https://orcid.org/0000-0002-6050-9867}{}
\author{Arturo De Faveri}{Université Paris Cité, CNRS, IRIF, F-75013, Paris, France}{defaveri@irif.fr}{https://orcid.org/0009-0004-0115-4500}{}
\author{Giulio Manzonetto}{Université Paris Cité, CNRS, IRIF, F-75013, Paris, France}{gmanzone@irif.fr}{https://orcid.org/0000-0003-1448-9014}{}
\author{Antonino Salibra}{Université Paris Cité, CNRS, IRIF, F-75013, Paris, France}{salibra@unive.it}{https://orcid.org/0000-0001-6552-2561}{}

\authorrunning{A. Bucciarelli, A. De Faveri, G. Manzonetto and A. Salibra} 

\Copyright{Antonio Bucciarelli, Arturo De Faveri, Giulio Manzonetto and Antonino Salibra} 


\nolinenumbers 

\EventEditors{Frank Pfenning}
\EventNoEds{1}
\EventLongTitle{11th International Conference on Formal Structures for Computation and Deduction (FSCD 2026)}
\EventShortTitle{FSCD 2026}
\EventAcronym{FSCD}
\EventYear{2026}
\EventDate{July 20--23, 2026}
\EventLocation{Lisbon, Portugal}
\EventLogo{}
\SeriesVolume{378}
\ArticleNo{14}

\begin{document}

\maketitle

\begin{abstract}
We study invertibility of $\lambda$-terms modulo $\lambda$-theories.
Here a fundamental role is played by a class of \lam-terms called finite hereditary permutations (FHP) and by their infinite generalisations (HP). 
More precisely, $\FHP$s are the invertible elements in the least extensional \lam-theory $\boldsymbol{\lambda \eta}$ and $\HP$s are those in the greatest sensible \lam-theory $\th H^*$.   
Our approach is based on inverse semigroups, algebraic structures that generalise groups and semilattices.
We show that FHP modulo a  $\lambda$-theory $\th T$ is always an inverse semigroup and that HP modulo $\th T$ is an inverse semigroup whenever
$\th T$ contains the theory of B\"ohm trees.
An inverse semigroup comes equipped with a natural order.
We prove that the natural order corresponds to $\eta$-expansion in $\FHP /\th T$, and to infinite $\eta$-expansion in $\HP/\th T$. 
Building on these correspondences we obtain the two main contributions of this work: 
firstly, we recast in a broader framework the results cited at the beginning; 
secondly, we prove that the $\FHP$s are the invertible \lam-terms in all the \lam-theories lying between $\boldsymbol{\lambda \eta}$ and $\th H^+$. 
The latter is Morris' observational \lam-theory, defined by using the $\beta$-normal forms as observables.
\end{abstract}

\section{Introduction}
\input{intro.tex}

\section{Preliminaries}\label{sec:prelims}
\input{prelim.tex}

\section{Permutation trees}\label{sec:permtrees} 
\input{pt.tex}

\section{Permutation trees and invertibility in $\lambda$-calculus}
\input{ptlambda.tex}

\section{Related and further work}
\input{conclusion.tex}

\enlargethispage{\baselineskip}
\bibliography{sample-base.bib}

\input{appendix.tex}

\end{document}

%% file: intro.tex
The \lam-calculus, established by Alonzo Church in the 1930s, has played a prominent role in theoretical computer science for almost a century~\cite{B84,BDS13,BM22}. 
This stems from its position at the crossroads of programming language theory and various branches of mathematics, which draws the attention of a diverse community of researchers. 
Its denotational semantics has been studied using domain theory~\cite{AmadioC98}, category theory~\cite{LambekS86,AspertiL91} and bicategory theory \cite{GuerrieriO21,KerinecMO23}, game semantics~\cite{GianantonioFH99,BlondeauPatissierCA25,Clairambault24}, linear logic~\cite{Abramsky93,Ehrhard02}, set theory~\cite{Krivine14,Diaz-CaroGMV19}, and universal algebra~\cite{MS08}.
Operational aspects of its semantics have been also mechanised in proof assistants~\cite{Czajka20,LancelotAV25}.

This article fits within the line of research that studies the \lam-calculus from an algebraic perspective. 
In~\cite{Church1941}, Church formulated the \lam-calculus as a semigroup where $M\circ N = \lam x.M(Nx)$, whence its connection with algebra begins at the syntactic level, 
and this perspective also inspired subsequent works \cite{BD89}. 
A problem that naturally arises is the characterisation of those \lam-terms $M$ that are \emph{invertible} in the sense that $M\circ N = N \circ M = \mathbf{I}$ (the identity), for some $N$. 
The set of invertible $\lambda$-terms modulo a $\lambda$-theory is a group. 
When the notion of equality is $=_\beta$ the group is trivial, as the only invertible \lam-term is $\mathbf{I}$. 
The problem becomes more interesting when considering extensional \lam-theories: Dezani-Ciancaglini, Bergstra and Klop have characterised the \lam-terms invertible modulo $=_{\beta\eta}$ as the \emph{finite hereditary permutations} (FHPs)~\cite{D76,BergstraK80}. 
Intuitively, FHPs are obtained from $\eta$-expansions of $\mathbf{I}$ by arbitrarily permuting the subtrees of their B\"ohm trees, and it can be proved that they are the only realisers of type isomorphisms in the simply typed \lam-calculus~\cite{Cosmo05}.
The invertibility modulo $=_{\beta\eta}$ has also been studied by Folkerts~\cite{Folkerts98}, who proved that a \lam-term is invertible if and only if it is both injective and surjective--a property that \emph{seems} natural but fails in other \lam-theories.
More recently, Statman has shown that the semigroup of $\lambda$-terms modulo $=_\beta$ is Sq-universal~\cite{Statman}, a profound result stating that every countable semigroup is embeddable in the semigroup of $\lambda$-terms modulo a suitable $\lambda$-theory.

There are other \lam-theories that are of interest for computer scientists because they capture operational properties of programs. 
In particular, Morris defined observational equivalences where two \lam-terms are equivalent if they display the same behaviour when they are plugged in any context~\cite{MorrisTh}. 
Clearly, this notion depends on the behavioural property one decides to observe.
The observational equivalence $\Hst$, where one observes head termination, is by far the most studied \lam-theory: it arises as the theory of Scott's model $\mathcal{D}_\infty$ and captures the B\"ohm tree equality modulo an infinitary notion of extensionality~\cite{Hyland75}. 
In \cite{BergstraK80}, Bergstra and Klop characterised the class of invertible \lam-terms modulo $\Hst$ by generalising the finite hereditary permutations to permutations having a possibly infinite B\"ohm tree ($\HP$s).

It is natural to consider the problem of characterising invertibility modulo \lam-theories lying between $=_{\beta \eta}$ and $\th H^*$; 
one of these is for instance $\th H^+$, Morris's observational theory in which the observed behavioural property is $\beta$-normalisation~\cite{IMP19}. 
In this paper, we aim to advance the study of this problem by extending the results presented in \cite[Chapter 21]{B84}.  
Our novel approach consists in showing that $\HP$, which is a monoid modulo $=_\beta$ 
and a group modulo $\th H^*$, is an \emph{inverse semigroup} (actually, an inverse monoid) 
modulo the \lam-theory of B\"ohm trees.
An inverse semigroup \cite{law98} is a semigroup in which every element admits a unique inverse in the generalised sense of semigroup 
theory.
By virtue of their rich structure, inverse semigroups are connected with various areas of logic and theoretical 
computer science, including topos theory \cite{Funk,FS10}, Stone duality \cite{L12}, and automata theory \cite{ARW}.
They simultaneously and faithfully generalise groups--since they are equipped with a unary operation extending 
the group inverse--and semilattices, as the idempotent elements of an inverse semigroup always form a 
semilattice with respect to a natural order. 
An inverse semigroup is a semilattice if and only if every element is idempotent, and it is a group if and only if the identity is the only idempotent.
The prototypical example of an inverse monoid is given by the injective partial endofunctions on a set; 
viewed in this way, inverse monoids provide a natural framework for describing partial symmetries, as opposed to the total ones described by groups.

We introduce a new monoid whose elements we call \emph{permutation trees}, which can be put into correspondence with a monoid of Böhm-like trees. 
We show that the permutation trees form an inverse monoid, we characterise its idempotents and the associated order,
and we prove that it is in fact $F$-inverse, a technical concept that will be appropriately explained. 
The analysis of this monoid has two important consequences: $\HP$ modulo Böhm tree equality and $\FHP$ modulo $=_\beta$ form inverse monoids;  
these sets remain so when considered modulo any theory extending Böhm tree equality and $=_\beta$, respectively.
A noteworthy aspect is that the natural order on these inverse monoids coincides with the preorder induced by infinite $\eta$-expansion (on Böhm-like trees) and finite $\eta$-expansion (on $\lambda$-terms). 
Quotienting an inverse monoid by a suitable congruence $\sigma$, called minimum group congruence, yields a group in which all idempotents lie in the equivalence class of the identity. 
In this group, the generalised inverses become actual inverses in the usual group-theoretic sense.
This naturally raises the question of what such a congruence corresponds to in the cases of $\HP$ modulo Böhm tree equality and $\FHP$ modulo $=_{\beta}$. 
The minimum group congruence is related to the natural order that corresponds to infinite $\eta$-expansion in the theory of Böhm trees and to finite $\eta$-expansion in $=_{\beta}$. 
Therefore, we obtain the key result that quotienting $\HP$ by Böhm theory and then by $\sigma$ gives precisely $\HP$ modulo $\th H^{*}$; 
similarly, quotienting $\FHP$ by $=_\beta$ and then by $\sigma$ gives precisely $\FHP$ modulo $=_{\beta \eta}$. 
This provides a new perspective on invertible $\lambda$-terms modulo these \lam-theories.

The $\lambda$-theories $\th H^*$ and $=_{\beta\eta}$ are special in that they arise precisely from the minimum group congruence, but this is not always the case. 
A counterexample is given by the theory $\th H^+$, whose equality has been characterised in terms of the interaction between Böhm trees and \emph{finite} $\eta$-expansions~\cite{CDZ,DRP,IMP19}. 
To characterise when two hereditary permutations are equal in $\th H^+$, we therefore need to take a further step and study the interaction between infinite trees equipped with an order that nonetheless originates from \emph{finite} idempotents.
In other words, we are led to consider a congruence inspired by, but distinct from, the minimum group congruence.
In the theory of Böhm trees, the order constructed in this way coincides with finite $\eta$-expansion. 
Building on the characterisation of the invertibles we gave for $\th H^*$, we demonstrate that the $\lambda$-terms invertible modulo $\th H^+$ are precisely $\FHP$s. 
As a consequence, $\FHP$s are the invertible $\lambda$-terms in all the \lam-theories lying between $=_{\beta \eta}$ and $\th H^+$.
This result settles a conjecture formulated by Barendregt in his book~\cite[p.~547]{B84}.
The conceptual reconstruction of invertibles modulo $\th H^*$ and $=_{\beta \eta}$ together with the characterisation of invertibles modulo $\th H^+$
constitute the main contributions of this work. 
Some directions for future research are outlined in the conclusion.

%% file: prelim.tex
We recall some basic notions and results concerning the $\lambda$-calculus and the theory of semigroups.
We refer to \cite{B84,BM22} for the former, and to \cite{Howie,law98} for the latter.

We denote by $\mathbb N$ the set of natural numbers and by $\mathbb N^+$ the set of
positive natural numbers.
For every $n \in \mathbb{N}$, let $S_n$ be the symmetric group on $n$ elements. 
Let $\mathrm{Sym}(\mathbb{N}^{+})$ be the symmetric group on $\mathbb{N}^{+}$. 
For every $n \ge 1$, $S_n$ can be embedded in $\mathrm{Sym}(\mathbb{N}^{+})$ as the subgroup of all the $\rho \in \mathrm{Sym}(\mathbb{N}^{+})$ such that $\rho i =i $ for all $i >n$.
By abuse of notation, we shall write $S_n$ for the image of this embedding.  
We denote by $\iota$ the identity permutation in $\mathrm{Sym}(\mathbb{N}^{+})$.

\subsection{Inverse semigroups}
Given a semigroup $S$,
we say that $u \in S$ is \emph{regular} if there is $v \in S$ such that $uvu=u$.  
A semigroup is called \emph{regular} if every element is regular. 
Moreover, we say that $v$ is an \emph{inverse} of $u$ if $uvu=u$ and $vuv=v$.
Every element in a regular semigroup admits at least on inverse.  
An extremely important class of regular semigroups is that of inverse semigroups, in which every element admits precisely one inverse.  
\begin{definition}
An \emph{inverse semigroup} is a semigroup $S$  endowed with a unary operation $(-)^*$ that satisfies, for all $u, v \in S$:
\begin{equation}
    \label{eq:inverse}
    (u^*)^*=u, \qquad (uv)^* = v^*u^*, \qquad uu^*u = u, \qquad uu^*vv^* = vv^*uu^* .
\end{equation}
An \emph{inverse monoid} is an inverse semigroup together with a unit for multiplication.
\end{definition}
In an inverse semigroup $u^*$ is the unique inverse of $u$.

\begin{example}
The prototypical example of an inverse monoid is the set $I(X)$ of partial injections on a given set $X$. 
An element $u$ is a bijection between two subsets of $X$, the domain and the range of $u$, and $u^*$ is its partial inverse. 
\end{example}

An \emph{idempotent} in a semigroup $S$ is $u \in S$ such that $uu=u$. 
The idempotents of $I(X)$ are the partial identities. 
We denote by $E(S)$ the set of idempotents of $S$. 
In any inverse semigroup, every element of the form $uu^*$ is idempotent. 
Conversely, since if $u$ is idempotent, then $u=u^*$, idempotents are precisely those elements of the form $uu^*$. 
The last identity in \eqref{eq:inverse} says that idempotents commute.

\begin{remark}\label{rem:surj-homo}
If $f: S \to T$ is a surjective homomorphism of semigroups and $S$ is an inverse semigroup, then $T$ is also an inverse semigroup. 
\end{remark}

\subparagraph*{The natural order}
Any inverse semigroup $S$ comes equipped with a natural order that corresponds to inclusion between functions of $I(X)$. 
By definition, $u \le v$ if $uu^* v = u$, or, equivalently, if $v u^*u = u$. 
Actually, one can prove that the following are equivalent: 
\begin{enumerate}
\item $u \le v$; 
\item $u=ve$, for some idempotent $e$; 
\item $u=dv$, for some idempotent $d$. 
\end{enumerate}
Therefore, if $d$ is idempotent and $u \le d$, then $u$ is idempotent.
Moreover, $(S,\le)$ is a partial order compatible with the structure of inverse semigroup, i.e., if $s \le t$ and $u \le v$, then $s^* \le t^*$ and $su \le tv$.
If $S$ is an inverse monoid, then the poset $(E(S), \le)$ has binary meets given by multiplication and a top element given by the unit of the monoid. 

\subparagraph*{The minimum group congruence}
Two important relations in the study of an inverse semigroup $S$ are:
\begin{enumerate}
    \item the relation of \emph{compatibility} $\sim$, defined by $t\sim u$ if $t^* u$ and $tu^*$ are idempotents; 
    \item the relation $\sigma_S$, defined by $t \,\sigma_S\, u$ if there exists $w$ such that $w\leq t$ and $w\leq u$. 
\end{enumerate}

When $t$ and $u$ are compatible, i.e., $t \sim u$, then $t \land u$ exists and $t \land u=tu^*u$. 
The relation $\sigma_S$ is the smallest congruence such that $S/\sigma_S$ is a group \cite[Thm. 2.4.1]{law98} and will be called \emph{minimum group congruence}.  

Two classes of inverse semigroups that have been studied extensively are the class of $E$-unitary inverse semigroups and the class of $F$-inverse semigroups. 
\begin{definition}
    An inverse semigroup $S$ is said to be 
    \begin{enumerate}
        \item \emph{$E$-unitary} if, whenever $e \in E(S)$ and $e \le v$, it follows that $v \in E(S)$;  
        \item \emph{$F$-inverse} if each equivalence class of $\sigma_S$ has a unique maximal element. 
    \end{enumerate}
\end{definition} 

The inverse semigroup $S$ is $E$-unitary iff ${\sim} = \sigma_S$ \cite[Thm. 2.4.6]{law98};  
if $S$ is $F$-inverse, then it is $E$-unitary \cite[Proposition 7.1.3]{law98}. 
The following result is \cite[Thm. 2.4.6]{law98}. 

\begin{lemma}
    \label{lem:mgc}
    The minimum group congruence $\sigma_S$ of an $E$-unitary inverse monoid $S$ is the least congruence equating all the idempotent elements of $S$. 
    Moreover, for every $u \in S$, $u \, \sigma_{S} \, 1$ iff $u \in E(S)$. 
\end{lemma}

\subparagraph*{Green's relations}
Certain equivalence relations, called Green's relations, have played a central role in the development of the theory. 
They can be defined on any semigroup, but assume a particularly simple form on inverse semigroups. 
\begin{definition}
    Let $S$ be an inverse semigroup. 
    We define two equivalence relations $\mathcal L$ and $\mathcal R$ on $S$ by $u \, \mathcal L \, v$ if $u^* u=v^* v$ and $u \, \mathcal R \, v$ if $u u^*=vv^*$. 
    Moreover, we define $\mathcal H := \mathcal L \cap \mathcal R$. 
\end{definition}
Traditionally, given $s \in S$, $L_s$, $R_s$, and $H_s$, denote their equivalence classes under $\mathcal{L}$, $\mathcal{R}$, and $\mathcal{H}$, respectively.
It is easy to check that $\mathcal L$ is a right congruence and $\mathcal R$ is a left congruence, 
meaning that if $u \, \mathcal{L} \, v$, then $us \, \mathcal{L} \, vs$ for every $s \in S$, and similarly for $\mathcal{R}$. 

\begin{proposition} \emph{\cite[Prop.~3.2.1]{law98}}
    \label{prop:key}
    Let $S$ be an inverse monoid and $e \in E(S)$. Then: 
    \begin{enumerate}
        \item $eSe$ is an inverse monoid, and is the largest monoid with identity $e$; 
        \item the $\mathcal{H}$-class $H_e$ is a group with identity $e$ and it coincides with $\{t \in eSe : tt^*=t^*t=e\}$; 
        \item every subgroup of $S$ is contained in an $\mathcal{H}$-class.
    \end{enumerate}
\end{proposition} 

Thus, for any idempotent $e \in E(S)$, the $\mathcal{H}$-class $H_e$ is the unique maximal subgroup of $S$ with identity $e$. 
Since each $\mathcal{H}$-class contains at most one idempotent, the maximal subgroups of $S$ are precisely the disjoint $\mathcal{H}$-classes containing an idempotent.

\subsection{Trees} 
An \emph{ordered tree} is a non-empty subset $A$ of $(\mathbb{N}^{+})^{*}$ such that:
\begin{enumerate}[(i)]
\item $A$ is closed under initial segments; 
\item $A$ is \emph{gapless}: for every $i,j \in \mathbb{N}^{+}$ and every $w \in A$, if $i < j$ and $wj \in A$, then $wi \in A$. 
\end{enumerate}
The elements of $A$ are called \emph{nodes}; 
a node is a \emph{leaf} if it has no children (immediate successors).
The empty word $\epsilon$ is the unique \emph{root}, i.e., node without immediate predecessors. 
We will be exclusively interested in \emph{finitely branching trees}, i.e., trees whose nodes have a finite number of children. 
We denote by $\rd(A)$ the \emph{degree} of the root of $A$, i.e., the number of its children.
For each $1\leq i\leq \rd(A)$ we denote by $A_i:= \{ w \in (\mathbb{N}^{+})^{*} : iw \in A \}$ the subtree of $A$ of root $i$.
Finally, we say that a \emph{labelled} tree is a function $l: A \to L$ from a tree to a nonempty set $L$ of labels. 
The tree $A$ is called the \emph{skeleton} of $l$, and is denoted by $\sk(l)$.

The set of finitely branching ordered trees can be defined coinductively.
When we deal with infinite trees we will freely use coinductive arguments.
We refer the reader to \cite{KS} for a gentle introduction to coinductive techniques.  

\subsection{Lambda calculus}
The set $\Lambda$ of \emph{$\lambda$-terms} over a countably infinite set $\mathrm{Var}$ of variables is defined by the grammar: 
\[
	M,N ::= x\mid \lam x.M\mid MN,\textrm{ for }x\in\mathrm{Var}.  
\]
We assume that application associates on the left, and has higher precedence than abstraction.
For instance, $\lam xyz.xyz$ stands for $\lam x.(\lam y.(\lam z.(xy)z))$.
From now on, $\lambda$-terms are considered up to renaming of bound variables. 
A $\lambda$-term without free variables is called a \emph{combinator}.
The \emph{$\beta$-reduction} $\to_\beta$ and \emph{$\eta$-reduction} $\to_\eta$ are defined as the contextual closures of the rules: 
\begin{equation*}
    (\lambda x.M)N \to_\beta M[N/x], \qquad \qquad \lambda x.Mx \to_\eta M, \,   \text{if $x$ does not occur free in $M$}
\end{equation*}
where we denote by $M[N/x$] the \emph{capture-free substitution} of $N$ for all free occurrences of $x$ in $M$. 
Given a notion of reduction $\to_r$ between $\lambda$-terms, we denote by $\twoheadrightarrow_r$ the reflexive, transitive closure of $\to_r$, and by $=_r$ the least equivalence relation containing $\to_r$. 
When $\to_r$ is confluent, $\nf_r(M)$ stands for the unique \emph{$r$-normal form} of $M$ (if any). 
We define \emph{$\beta\eta$-reduction} as $\to_{\beta \eta}\ =\ \to_\beta \cup \to_\eta$. 
    The combinators below are used throughout the paper: 
    \begin{equation*}
        \mathbf{I}= \lambda x.x, \qquad \mathbf{1}=\lambda xy.xy, \qquad \mathbf{B} = \lambda fgx. f(gx),\qquad \mathbf{Y} = \lam f.(\lam x.f(xx))(\lam x.f(xx)).
    \end{equation*}
We recall that $\mathbf{I}$ is the identity, $\mathbf{1}$ the 1st Church numeral, $\mathbf{B}$ stands for the composition and $\mathbf{Y}$ is Curry's fixed point combinator satisfying $\mathbf{Y}M=_\beta M(\mathbf{Y}M)$, for every $M \in \Lambda$.

\subparagraph*{B\"ohm trees}
A $\lambda$-term is in \emph{head normal form} (\emph{hnf}) if it is of the form $\lambda x_1 \ldots x_n.yM_1 \cdots M_k$ for some $n,k \in \mathbb{N}$. 
If a $\lambda$-term has a hnf, this hnf can be reached by head reductions $\to_h$, i.e., by repeatedly $\beta$-reducing its head redex. 
A $\lambda$-term with a hnf is called \emph{solvable} and a $\lambda$-term without a hnf is called \emph{unsolvable}.
    The \emph{B\"ohm tree} $\BT(M)$ of a \lam-term $M$ is a labelled ordered tree defined coinductively as follows: 
        if $M$ is unsolvable, then $\BT(M)= \bot$; 
        otherwise $M \twoheadrightarrow_h \lambda x_1 \ldots x_n.yM_1 \cdots M_k$, and
\[
            \begin{tikzpicture}[level distance=1cm, sibling distance=1cm]
                \node {$\BT(M) := \lambda x_1 \ldots x_n.y \hspace{3.1cm}$}
                child { node {$\BT(M_1)$} }
                child { node {$\cdots$} }
                child { node {$\BT(M_k)$} };
            \end{tikzpicture}
\]
The set $\mathscr{B}$ of \emph{B\"ohm-like trees} is coinductively defined as: $\bot \in \mathscr{B}$ and if $T \in \mathscr{B}$, 
then there are $k,n \in \mathbb{N}$, $x_1, \ldots, x_n, y \in \mathrm{Var}$, and $T_1, \ldots, T_k \in \mathscr{B}$ such that $T$ is 
 \begin{equation*}
            \begin{tikzpicture}[level distance=1cm, sibling distance=1cm]
                \node {$\lambda x_1 \ldots x_n.y \hspace{1.5cm}$}
                child { node {$T_1$} }
                child { node {$\cdots$} }
                child { node {$T_k$} };
            \end{tikzpicture}
        \end{equation*}
When $T \in \mathscr{B}$ is a finite tree without $\bot$, we identify $T$ with the corresponding $\lambda$-term.  
Given $T \in \mathscr{B}$, there exists $M \in \Lambda$ such that $\BT(M)=T$ exactly when $T$ is recursively enumerable and the number of free variables in $T$ is finite \cite[Thm. 10.1.23]{B84}.

\subparagraph*{Theories}
A \emph{$\lambda$-theory} is an equivalence relation $\th T \subseteq \Lambda^2$ on the set of $\lambda$-terms that contains $=_\beta$ and is compatible with abstraction and application.
The minimal $\lambda$-theory $=_\beta$ is also denoted by $\boldsymbol{\lambda}$. 
Given a $\lambda$-theory $\th T$, it is customary to write $\th T \vdash M = N$ for $M \, {\th T} \,  N$. 

A $\lambda$-theory $\th T$ is called: 
\begin{itemize}
\item \emph{consistent} if it does not equate all $\lambda$-terms; 
\item \emph{extensional} if it contains $=_{\beta \eta}$, or, equivalently, if $\th T \vdash \mathbf{I}=\mathbf{1}$; 
\item \emph{semisensible} if it does not equate a solvable and an unsolvable; 
\item \emph{sensible} if it is consistent and equates all the unsolvables.
\end{itemize}
The set of $\lambda$-theories ordered by inclusion forms a complete lattice having a rich structure~\cite{LusinS04}. 

For every theory $\th T$, we denote by $\th T \boldsymbol{\eta}$ the least extensional $\lambda$-theory containing $\th T$, 
and by $\th T {\boldsymbol{\omega}}$ the closure of $\th T$ under the so-called \emph{$\omega$-rule}: for all $M,N \in \Lambda$, if $\th T \vdash MC = NC$ for every combinator $C$, then $\th T \vdash M = N$.
The following are the most important \lam-theories:
\begin{itemize}
\item the least sensible $\lambda$-theory $\th H$. 
\item the \lam-theory ${\th B}$ induced by the B\"ohm tree equality: $M \, {\th B} \, N$ iff $\BT(M) = \BT(N)$.
\item Morris' observational theory $\th H^+$ defined by $\th H^+ \vdash M = N$ if, for all contexts $C[\,]$,
\begin{equation*}
\text{$C[M]$ has a $\beta$-nf iff $C[N]$ has a $\beta$-nf.}
\end{equation*}
\item the greatest sensible \lam-theory $\th H^*$ defined by $\th H^* \vdash M = N$ if, for all contexts $C[\,]$,
\begin{equation*}
\text{ $C[M]$ has a hnf iff $C[N]$ has a hnf.}
\end{equation*}
\end{itemize}
It can be proved that a \lam-theory $\th T$ is semisensible iff $\th T \subseteq \th H^*$. 
Since $\th B$ and $\th H^+$ are sensible, and $\th H^+$ is moreover extensional, we have $\boldsymbol{\lambda} \subseteq \th H \subseteq \th B \subseteq \th H^+ \subseteq \th H^*$, and thus $\boldsymbol{\lambda \eta} \subseteq \th H \boldsymbol{\eta} \subseteq \th B \boldsymbol{\eta} \subseteq \th H^+ \subseteq \th H^*$. By \cite[Thm. 17.4.16]{B84} and \cite{IMP17}, these inclusions are \emph{strict}.

\subparagraph*{Hereditary permutations}
A \emph{finite $\eta$-expansion} ($\mathcal{I}^{\eta}$) of $\mathbf I$ is any combinator $M$ such that $M \twoheadrightarrow_{\beta \eta} \mathbf I$.
Every $M \in \mathcal{I}^{\eta}$ has a $\beta$-normal form, that can be taken as a representative of $M$.
There is a bijection between $\beta$-normal forms of finite $\eta$-expansions of $\mathbf I$ and finite unlabelled ordered trees \cite{IN03}. 

\begin{example}
    The following $\lambda$-terms are finite $\eta$-expansions of $\mathbf I$:  
    \begin{itemize}
        \item $\mathbf{1}_n :=\lambda x x_1 \ldots x_n . x x_1 \cdots x_n$; 
        \item $\mathbf{1}^0:=\mathbf{I}$, $\mathbf{1}^{n+1}:=\lambda xz. x(\mathbf{1}^n z)$. 
    \end{itemize}
    The skeleton of $\BT(\mathbf{1}^k)$ is $\{1^n : n \le k\}$ and the union of these trees is  
the skeleton of the B\"ohm tree of Wadsworth's combinator $\mathbf J:= \mathbf Y (\lambda f g x. g (fx))$: $\sk(\BT(\mathbf{J}))=\{1^n : n \in \mathbb{N}\}$;  
$\lambda$-terms like $\mathbf J$ are called `infinite $\eta$-expansions' of the identity. 
This concept will be defined shortly.  
\end{example}

For every $n \in \mathbb{N}$, let $S_n$ be the symmetric group on $n$ elements. 
For all variables $x$, the set $\mathscr{H\!\!P}(x) \subseteq \mathscr{B}$ is defined coinductively as follows. 
If $T \in \mathscr{H\!\!P}(x)$, then there are $n \in \mathbb{N}$, $\pi \in S_n$ and $T_i \in \mathscr{H\!\!P}(x_{\pi i})$, $1 \le i \le n$ such that $T$ is 
 \begin{equation}
            \begin{tikzpicture}[level distance=1cm, sibling distance=1cm]
                \node {$\lambda x_1 \ldots x_n.x \hspace{1.5cm}$}
                child { node {$T_1 $} }
                child { node {$\cdots$} }
                child { node {$T_n $} };
            \end{tikzpicture}
        \end{equation}

    \begin{definition}
        \label{def:ietaomega}
    We define:
    \begin{itemize}
    \item $\mathscr{H\!\!P}:=\{\lambda x. T: T \in \mathscr{H\!\!P}(x)\}$; 
    \item $\HP := \{M \in \Lambda : \BT(M) \in \mathscr{H\!\!P}\}$; 
    \item $\FHP :=\{M \in \Lambda : \BT(M) \in \mathscr{H\!\!P} \text{ and } \BT(M) \text{ is finite}\}$.
    \end{itemize} 
An element of $\HP$ is called a \emph{hereditary permutation} and an element of $\FHP$ a \emph{finite hereditary permutation}. 
    \end{definition}

The set $\mathscr{I}(x) \subseteq \mathscr{H\!\!P}(x)$ is defined coinductively as follows.
If $T \in \mathscr{I}(x)$, then there are $n \in \mathbb{N}$, $T_i \in \mathscr{I}(x_{i})$, $1 \le i \le n$ such that $T$ is (2). 

\begin{definition}
    We define: 
    \begin{itemize}
        \item $\mathscr{I}:=\{\lambda x. T: T \in \mathscr{I}(x)\}$; 
        \item $\mathcal{I}^{\eta}_{\omega}:=\{M \in \Lambda : \BT(M) \in \mathscr{I}\}$. 
    \end{itemize} 
An element of $\mathcal{I}^{\eta}_{\omega}$ is called a \emph{possibly infinite $\eta$-expansion} of $\mathbf I$.  
\end{definition} 
When $M \in \mathcal{I}^{\eta}_{\omega}$ has a finite B\"ohm tree, we obtain a finite $\eta$-expansion of $\mathbf I$ as previously defined.

\subparagraph*{Invertibility}
The operation $M \circ N := \mathbf B MN$ on $\Lambda/\boldsymbol{\lambda}$ is associative, hence $(\Lambda/\boldsymbol{\lambda}, \circ)$ is a semigroup.
For a $\lambda$-theory $\th T$, a term $M \in \Lambda$ is \emph{$\th T$-invertible} if $\th T \vdash M \circ N = N \circ M = \mathbf I$ for some $N \in \Lambda$. 
This leads to the problem of characterising invertible $\lambda$-terms in a given theory; we summarise the known results, referring to \cite[Chapter 21]{B84}.
\begin{theorem}
    \label{thm:inv} Let $\th T$ be a $\lambda$-theory. 
    \begin{enumerate}
        \item If $\th T$ is $\boldsymbol{\lambda}$, $\th H$ or $\th B$, then $M$ is $\th T$-invertible iff $\th T \vdash M = \mathbf I$; 
        \item if $\th T$ is $\boldsymbol{\lambda \eta}, \boldsymbol{\lambda \omega}$, $\th H\boldsymbol{\eta}, \th H\boldsymbol{\omega}$, then $M$ is $\th T$-invertible iff $M \in \FHP$;
        \item if $\th T = \th H^*$, $M$ is $\th T$-invertible iff $M \in \HP$. 
    \end{enumerate}
\end{theorem}

%% file: pt.tex
In this section we introduce and study a new inverse monoid $\PT$ of labelled ordered trees,
that, as shown in Proposition \ref{prop:embed}, is isomorphic to the monoid of hereditary permutations $(\mathscr{H\!\!P}, \circ, \mathbf I)$. 
Working with $\PT$ has the advantage of revealing the underlying structure in a more transparent way.

\begin{definition}
    A \emph{permutation tree} is a labelled ordered tree $t: A \to \mathrm{Sym}(\mathbb{N}^{+})$ such that: 
    \begin{itemize}
        \item if $a \in A$ has exactly $n \ge 1$ children, then $t(a)\in S_n$; 
        \item if $a \in A$ is a leaf, then $t(a):=\iota$, where $\iota$ is the identity permutation in $\mathrm{Sym}(\mathbb{N}^{+})$. 
    \end{itemize}  
    We let $\PT$ be the set of permutation trees. 
\end{definition}

We denote by $\bar\epsilon$ the unique permutation tree with skeleton $\{\epsilon\}$, i.e.\ the tree with one node labelled $\iota$. 

We can conveniently write a permutation tree $t$ as $\< \pi; t_1, \ldots, t_n\>$ where $\pi \in S_n$ is the label of the root, $\rd(t)=n$ and $t_1, \ldots, t_n \in\PT$.

Observe that $\PT$ can be defined coinductively as the greatest set $\PT$ such that, if $t \in \PT$, then there is $n \in \mathbb N$, $\pi \in S_n$ and $t_1, \ldots, t_n \in\PT$ such that $t=\<\pi; t_1, \ldots, t_n\>$. 

\begin{definition}
 Let $t, u \in \PT$. 
We define the product $tu$ by coinduction as follows. 
First, $t \bar \epsilon = t = \bar\epsilon t$. 
If $t = \<\pi ;t_1,\ldots,t_n\>$, $u = \<\rho;u_1,\ldots,u_m\>$ and $k = \max\{n,m\}$, then we define:
\begin{equation*}
    tu := \<\pi \circ \rho; u_1 t_{\rho1} ,\ldots, u_k t_{\rho k} \> 
\end{equation*}
where, by convention, $t_j =\bar \epsilon$ if $j>n$, and $u_i =\bar \epsilon $ if $i>m$.
\end{definition}

The product operation makes $(\PT,\cdot,\bar \epsilon)$ into a monoid. 
We now show that this monoid is an inverse monoid. 

\begin{definition}
 Given $ t\in\PT$, we define by coinduction $t^*$ as follows. 
 First, $\bar\epsilon^*=\bar\epsilon$. 
 If $t=\<\pi; t_1, \ldots, t_n\>$ with $n \ge 1$, then $t^* = \<\pi^{-1}; (t_{\pi^{-1}1})^*, \ldots, (t_{\pi^{-1}n})^*\>$. 
\end{definition}

We denote by $E$ the set $\{t \in \PT: t(a)=\iota \text{ for every } a \in \sk(t)\}$. 
Every $t \in E$ is univocally determined by its skeleton and we identify $\sk(t)$ with the permutation tree $t$. 

\begin{proposition}
    \label{lem:idemp}
    The set $E(\PT)$ of idempotent elements of $\PT$ coincides with $E$. 
    Moreover, for every $d,e \in E$, $d e=d\cup e$; consequently any two idempotents commute.  
\end{proposition}
\begin{proof}
    Let $t=\<\pi; t_1, \ldots, t_n\> \in E(\PT)$ be idempotent. 
    Then $\pi \circ \pi = \pi$, and this is possible only if $\pi = \iota$. 
    The claim $t \in E$ follows from $t_i = (tt)_i = t_it_i$.
    Conversely, if $d = \<\iota;d_1,\ldots,d_n\> \in E$, then by definition of the product $dd = \<\iota;d_1d_1,\ldots,d_nd_n\>$. 
    Since by coinduction $d_i d_i = d_i$, we have $dd = d$ and $E = E(\PT)$.
    Finally, if $d,e \in E$, then, again by the coinductive definition of the product, $de=d \cup e$. 
\end{proof}

\begin{theorem}
    \label{thm:perminverse} 
    The monoid $(\PT,\cdot,(-)^*,\bar\epsilon)$ is an $E$-unitary inverse monoid. 
\end{theorem}
\begin{proof}
    We prove that $\PT$ is an inverse monoid by applying \cite[Theorem 5.1.1]{Howie}, which states that $S$ is an inverse monoid iff $S$ is a regular monoid whose idempotents commute.  
    For every $t=\<\pi; t_1, \ldots, t_n\> \in \PT$, we prove that $ t  t^*  t= t$:  $\pi \circ \pi^{-1} \circ \pi = \pi$ and, coinductively, 
\begin{align*}
    (t  t^*  t)_i & =(( t  t^*)  t)_i 
     = t_i  ( t  t^*)_{ \pi i} 
     = t_i   (t^*)_{ \pi i}  t_{ (\pi^{-1} \circ \pi)i} 
     = t_i   ( t_i)^*  t_i 
     =t_i\text{.}
\end{align*}
Since for every $t$ there is $t^*$ such that $t t^* t= t$, $\PT$ is regular.
We conclude that $\PT$ is an inverse monoid applying Proposition \ref{lem:idemp}.

We now prove that $\PT$ is $E$-unitary. 
Let $t=\<\pi;t_1, \ldots, t_n\>$, and $e = \<\iota; e_1, \ldots,e_m\>$ be an idempotent 
such that  $e\leq t$. 
If $e=\bar\epsilon$, then necessarily $t=\bar\epsilon$, hence $t \in E$.
If $e \neq \bar\epsilon$, then $e=td$ for some $d \in E$. 
We prove that $t \in E$. 
First, we immediately get that $\pi=\iota$ and $n\leq m$. 
Moreover, for every $1\leq i\leq n$, we have $e_i=d_i t_i$. 
Since $d_i$ is idempotent, we get that $e_i\leq t_i$, and coinductively, that  $t_i \in E$. 
We conclude that $t \in E$. 
\end{proof}

As a consequence of $E$-unitarity, if $t\leq u$ and one between $t$ and $u$ is idempotent, then both are idempotent.

We now show that $\PT$ is isomorphic to a monoid of Böhm-like trees. 
The set of B\"ohm-like trees is endowed with a composition that turns $(\mathscr{B}, \circ)$ into a semigroup \cite[Thm.~18.3.10]{B84}. 

\begin{proposition}
    \label{prop:embed}
    There is an embedding $(-)^+$ of the inverse monoid $\PT$ of permutation trees into the semigroup $\mathscr{B}$ of B\"ohm-like trees 
    whose image is $\mathscr{H\!\!P}$. 
    In particular, $(\mathscr{H\!\!P}, \circ, \mathbf I)$ is an inverse monoid.
\end{proposition}
\begin{proof}
        We coinductively define an injective map $(-)^+ : \PT \to \mathscr{B}$.
        We exploit the fact that $\mathscr{B}$ is a lambda abstraction algebra \cite[Lemma~23]{S00},
        which entails that application and abstraction of Böhm-like trees are computed in the same way as for \lam-terms. 
        We define: 
        \begin{itemize}
            \item $t^+:=\lambda x.x$ if $t= \bar \epsilon$; 
            \item $t^+:=\lambda x x_1 \ldots x_n. x ((t_1)^+ x_{\pi 1}) \cdots ((t_n)^+ x_{\pi n})$ if $t=\<\pi;t_1, \ldots, t_n\>$. 
        \end{itemize}
We now prove that $(tu)^+= t^+\circ u^+$.
Let $u=\<\rho; u_1, \ldots, u_m\>$.
Without loss of generality, we assume $n=m$: 
\begin{align*}
    (tu)^+ 
    & = \lambda xx_1\dots x_n.x((u_1t_{\rho 1})^+ x_{(\pi \circ\rho)(1)}) \cdots((u_n t_{\rho n})^+x_{(\pi \circ \rho)(n)})\\ 
    & = \lambda xx_1\dots x_n.x(((u_1)^+\circ(t_{\rho 1})^+)x_{(\pi \circ \rho )(1)})\cdots(((u_n)^+\circ(t_{\rho n})^+)x_{(\pi\circ\rho)(n)}) \\
    & = \lambda xx_1\dots x_n.u^+ x((t_1)^+x_{\pi 1})\cdots((t_n)^+x_{\pi n}) \\
    & = t^+ \circ u^+. 
\end{align*}
The last statement of the proposition follows from Remark \ref{rem:surj-homo}.
\end{proof}

Note that, for $A \in \mathscr{H\!\!P}$, $A^*$ coincides with the tree $A^{-1}$ of \cite[Def. 21.2.16]{B84}.
The image of the idempotents of $\PT$ is the set $\mathscr{I}$ introduced in Definition \ref{def:ietaomega}.

\subsection{The minimum group congruence}
In this section we prove that the monoid $\PT$ is $F$-inverse. 
In the first lemma we characterise 
the natural partial order of the inverse monoid $\PT$. 

\begin{lemma}\label{lem:order}  Let $t=\<\pi;t_1, \ldots, t_n\>$ and $u=\<\rho; u_1, \ldots, u_m\>$. 
    Then $t \le u$ iff 
    \begin{enumerate}
    \item $\pi = \rho$ and $m \le n$; 
    \item $t_i \le u_i$ for $1 \le i\le m$ and $t_i \in E$ for $m < i \le n$. 
\end{enumerate} 
\end{lemma}
\begin{proof}
        Assume that $t \le u$, meaning that $t=ue$ for some $e \in E$. 
    Then $\pi = \rho$ and $n=\max\{m,\rd(e)\}\geq m$. 
    Moreover, since $t_i = e_i u_i$, for $1\leq i\leq m$, and $t_i=e_i$, for $m< i\leq n$, the conclusion easily follows. 
    Conversely, let $e=\<\iota;e_1, \ldots, e_m, t_{m+1}, \ldots, t_n\>$, where $e_i$ is such that $t_i=e_i u_i$ with $e_i \in E$ for $1 \le i \le m$. 
    Then we have $t=ue$. 
\end{proof}

Remark that the natural order is the greatest relation on $\PT$ that satisfies
conditions (1) and (2) of Lemma \ref{lem:order}.

\begin{definition}
    \label{def:max}
  Given $t=\<\pi; t_1, \ldots, t_n\> \in\PT$, we define coinductively $t^\mathfrak{m}$ as follows: 
\begin{equation*}
    t^\mathfrak{m}:=
    \begin{cases}
        (\<\pi; t_1, \ldots, t_{n-1}\>)^{\mathfrak{m}}, & \text{ if } t_n \in E \text{ and } \pi n = n, \\
        \<\pi; (t_1)^{\mathfrak{m}}, \ldots, (t_n)^{\mathfrak{m}}\> ,& \text{ otherwise.}
    \end{cases}
\end{equation*}  
\end{definition}

Observe the similarity with the definition of Nakajima tree of a $\lambda$-term given in \cite[Definition 2.60]{BM22}. 
See also \cite[Proposition 10.2.15]{B84}. 
Remark that if $e \in E$, then $e^{\mathfrak{m}}=\bar \epsilon$.
For $t \in\PT$, let $t^\uparrow:=\{u : t \le u\}$. 

\begin{proposition}
    \label{lem:ume} For every $t,u \in\PT$, the following hold: 
    \begin{enumerate}
        \item $t \le t^{\mathfrak{m}}$; 
        \item if $t \le u$, then $u \le t^{\mathfrak{m}}$. 
    \end{enumerate}
    In particular, $t^{\mathfrak{m}}$ is the unique maximal element in $t^\uparrow$. 
\end{proposition}
\begin{proof}
    If $t$ is idempotent, (1) is trivial and (2) follows from $E$-unitarity. 
    Otherwise let $t=\<\pi; t_1, \ldots, t_n\>$ and let $k$ be the greatest $i$ such that $t_k \notin E$ or $\pi k \neq k$.
    Then $t^\mathfrak{m}=\<\pi; (t_1)^{\mathfrak{m}}, \ldots, (t_k)^{\mathfrak{m}}\>$. 
    By coinductive hypothesis $ t_i \le (t_i)^{\mathfrak{m}}$ for $1 \le i \le k$, and $t_i \in E$ for $k < i \le n$, so that by Lemma \ref{lem:order} $t \le t^{\mathfrak{m}}$, showing (1).  
    Concerning (2), let $u=\<\rho; u_1, \ldots, u_m\>$ such that $t \le u$.
    Without loss of generality, we assume that $n=m$, so that $k \le m$, where $k$ is defined above. 
    Moreover, we have $\pi = \rho$, so that by Lemma \ref{lem:order} and 
    by coinductive hypothesis, $u_i \le (t_i)^{\mathfrak{m}}$ for $1 \le i \le k$. 
    As $t_i \le u_i$ and $t_i \in E$ for $k< i \le n$, we have $u_i \in E$ for $k< i \le m$ by $E$-unitarity. 
    By Lemma \ref{lem:order} we conclude that $u \le t^{\mathfrak{m}}$. 
\end{proof}

We now prove that $\PT$ is $F$-inverse. The proof relies on the following general lemma that will be used multiple times.  

\begin{lemma}\label{lem:max} Let $S$ be a $E$-unitary inverse semigroup. 
    Assume that for every $t \in S$, the set $t^{\uparrow}$ has a unique maximal element $t^{\mathfrak{m}}$. 
    Then for all $t,u \in S$, $t \, \sigma_S\,  u$ iff $t^{\mathfrak{m}}=u^{\mathfrak{m}}$. 
\end{lemma}

\begin{proof} 
Assuming $t \, \sigma_S\,  u$, there exists $w$ such that $w\leq t,u$. 
Hence $w^{\mathfrak{m}}=t^{\mathfrak{m}}=u^{\mathfrak{m}}$.  
Conversely, let $w=t^{\mathfrak{m}}=u^{\mathfrak{m}}$. 
Then $u=we$ for some $e\in E(S)$ and $t=dw$ for some $d\in E(S)$. 
Moreover, $t^* u=t^* w e\leq w^* we\in E(S)$ and $tu^*=dwu^*\leq dww^*\in E(S)$, so that $t,u$ are compatible, i.e., $t \sim u$.  
As $S$ is $E$-unitary, ${\sim} = \sigma_S$, so that $t \, \sigma_S\,  u$.   
\end{proof}

\begin{theorem}
The monoid $\PT$ is $F$-inverse.
\end{theorem}
\begin{proof}
 By Lemma \ref{lem:max} and Proposition \ref{lem:ume}.
\end{proof}

\subsection{Finite permutation trees}
We now focus on the subset of finite permutation trees ($\FPT$). 
With the operations inherited from $\PT$, this is an inverse monoid whose idempotent elements are the finite trees in $E$. 

\begin{definition}
    We denote by $E_\fin$ the set of finite permutation trees in $E$. 
\end{definition}

\begin{lemma}
    The monoid $\FPT$ is an $F$-inverse monoid. 
\end{lemma}
\begin{proof}
    The class of $F$-inverse monoids is an equational class when adding $(-)^{\mathfrak{m}}$ as a unary operation to the signature (see \cite{AKS}). 
    As $\FPT$ is a submonoid of $\PT$ closed under $(-)^{\mathfrak{m}}$, $\FPT$ is $F$-inverse. 
\end{proof}

The minimum group congruence $\sigma_\FPT$ is obtained restricting $\sigma_{\PT}$ to $\FPT$.  
\begin{lemma}
    \label{lem:res}
    For all $t,u \in\FPT$, $t \, \sigma_{\FPT} \, u$ iff $t \, \sigma_{\PT} \, u$. 
\end{lemma}
\begin{proof}
    One direction is trivial; the other follows from $E$-unitarity: if $t \, \sigma_{\PT} \, u$, then $t \sim u$, so that $t \land u$ exists and is given by $tu^*u$. 
    Since $tu^* u \in \FPT$, we have $t \, \sigma_{\FPT} \, u$. 
\end{proof}

Finite permutation trees have a rich structure in their own right. 
First, we show that $\FPT$ is residually finite. 
Second, we characterise the covering relation in the natural order, a result that will be useful later. 

\subparagraph*{The monoid of finite permutation trees is residually finite}
If $t \in\FPT$ with $\sk(t)=e$, we denote by $\norm t$ the number of nodes of $e$.
\begin{lemma}\label{lem:norm} For any $t,u \in\FPT$, the following hold: 
\begin{enumerate}
    \item if $t < u$, then $\norm u < \norm t$;
    \item $\norm t=\norm{t^*}$; 
    \item $\norm t, \norm u \le \norm{tu} \le \norm t + \norm u$; 
    \item $\norm{\sk(t)}=\norm{tt^*}=\norm{t^* t}$.
\end{enumerate}
As a consequence of \emph{(1)} and \emph{(3)} the sets $t^\uparrow$ and $\{(v,w): t=vw\}$ are finite. 
\end{lemma}

Two easy, yet remarkable, consequences of the previous lemma are the following facts.

\begin{proposition}
    Every finitely generated submonoid of $\FPT$ is finite. 
\end{proposition}

\begin{proposition}\label{lem:orderb} For every $t\in\FPT$, the equivalence classes $L_t$, $R_t$, $H_t$
of Green's relations are all finite sets. 
\end{proposition}
\begin{proof}
    By induction over the complexity of $t$.
    It $t=\bar\epsilon$, then $L_{\bar\epsilon}=R_{\bar\epsilon}=\{\bar\epsilon\}$.
    Otherwise, let $e= t t^*$. By definition, $t \,\mathcal R \, u$ iff $e=u u^*$. 
    Since $\{u : e=u u^*\}$ is finite by Lemma \ref{lem:norm}, we get the conclusion. 
    A similar proof works for $\mathcal L$.
\end{proof}

\begin{definition}
    A monoid $M$ is \emph{residually finite} if for every $u,v \in M$ there is a finite monoid $N$ and a homomorphism $\alpha: M \to N$ such that $\alpha(u) \neq \alpha(v)$. 
    We say that $N$ \emph{separates} $u$ and $v$. 
\end{definition}

\begin{theorem}
    The inverse monoid $\FPT$ is residually finite. 
\end{theorem}
\begin{proof}
    The conclusion follows from the fact that every 
    $\mathcal{R}$-class is finite by Proposition \ref{lem:orderb}, using 
    \cite[Lemma 5.3]{JM}. 
\end{proof}

Every residually finite monoid $S$ can be equipped with a natural topology, the so-called \emph{profinite} topology. 
This topology is defined by taking as a subbasis the equivalence classes of the congruences of finite index, that is, 
those congruences $\theta$ such that $S /\theta$ is finite. 
It is a metric topology with distance given by $d(u,v) = 2^{-r(u,v)}$ for all $u \neq v$, where $r(u,v)$ is the minimum cardinality of a monoid $M$ separating $u$ and $v$. 
With this topology, the monoid $S$ embeds into its profinite completion $\hat{S}$, which is obtained as the inverse limit of all its finite quotients. 
The monoid $\FPT$ is residually finite and therefore admits such a profinite topology; 
this topology turns out to be discrete.  
The proof of this fact, however, relies on technical arguments that go beyond the scope of the present article.
We leave the study of the completion of $\FPT$ for future work.

\subparagraph*{The covering relation}
Given two distinct $t$ and $u$ in $\FPT$, we write $t\prec u$ and say that $u$ \emph{covers} $t$ if $t< u$ and there is no $z$ such that $t<  z<  u$. 
The next two lemmas are proven by structural induction on $\FPT$s. 
The covering relation is characterised as follows.

\begin{lemma}\label{lem:singleton2} Let $t=\<\pi; t_1, \ldots, t_n\>$ and $u=\<\rho; u_1, \ldots, u_m\>$ be finite permutation trees. 
Then $t\prec u$ iff one of the following conditions holds:
	\begin{enumerate}
	\item $\pi = \rho$, $n=m$, and there is $k \in \mathbb N$ such that $t_k\prec u_k$ and $t_i= u_i$ for every $1\leq i\neq k\leq m$;
	\item $\pi = \rho$, $n=m+1$, $t_n=\bar\epsilon$ and $t_i= u_i$ for every $1\leq i\leq m$.
	\end{enumerate}
\end{lemma}

\begin{lemma}\label{lem:leqprec} 
 Let $t,u \in\FPT$ such that $t<  u$. 
 Then there exist $n \in \mathbb{N}$ and $z_0,z_1,\dots,z_{n+1} \in \FPT$
 such that $z_0=t\prec z_1\prec\dots\prec z_{n}\prec z_{n+1}=u$. 
\end{lemma}
\begin{proof}
    The conclusion follows because by Lemma \ref{lem:norm} $t^{\uparrow}$ is finite. 
\end{proof}

We shall see that the covering relation, modulo the embedding of Proposition \ref{prop:embed}, corresponds to $\eta$-reduction.

\subsection{Permutation trees and finite idempotents}
In the next section, the interaction between infinite permutation trees and idempotent $\FPT$s will be fundamental. 
This leads us to study a new order on $\PT$ inspired by the natural order and to investigate its maximal elements.
Recall that $E_\fin=E(\FPT)$. 

\begin{definition} We define ${\le_{\fin}}$ as the greatest relation on $\PT$ such that if $t \le_{\fin} u$ with 
    $t=\<\pi;t_1, \ldots, t_n\> \text{ and } u=\<\rho; u_1, \ldots, u_m\>\text{,}$ 
    then:
\begin{enumerate}
    \item $\pi = \rho$ and $m \le n$; 
    \item $t_i \le_{\fin} u_i$ for $1 \le i\le m$ and $t_i \in E_\fin$ for $m < i \le n$.  
\end{enumerate} 
\end{definition}

By coinduction, it is not difficult to see that this is a compatible partial order. 
Moreover, as $E_\fin \subseteq E$, if $t \le_{\fin} u$ then $t \le u$. 
In the next lemma we establish some properties of the order ${\le_{\fin}}$. 
The last one asserts some kind of unitarity with respect to $E_\fin$. 

\begin{lemma}
    \label{lem:oprime}
    Let $t,u \in\PT$. Then the following conditions hold: 
    \begin{enumerate}
        \item \label{lem:order-prime} there is $e \in E_\fin$ such that $t=eu$ iff there is $f \in E_\fin$ such that $t=uf$;  
        \item \label{lem:order-prime2} if there is $e \in E_\fin$ such that $t=ue$, then $t \le_{\fin} u$;
        \item  \label{lem:funi} if $f \le_{\fin} t$ with $f \in E_\fin$, then $t \in E_\fin$.
    \end{enumerate}
\end{lemma}

\begin{definition}
    \label{def:maxprime}
    Let $t=\<\pi; t_1, \ldots, t_n\> \in\PT$. 
    We define coinductively $t^{\mathfrak{m}'}$ as: 
\begin{equation*}
    t^{\mathfrak{m}'}:=
    \begin{cases}
        (\<\pi; t_1, \ldots, t_{n-1}\>)^{\mathfrak{m}'}, & \text{ if } t_n \in E_\fin \text{ and } \pi n = n,\\
        \<\pi; (t_1)^{\mathfrak{m}'}, \ldots, (t_n)^{\mathfrak{m}'}\>, & \text{ otherwise.}
    \end{cases}
\end{equation*}
\end{definition}
\begin{proposition}
    \label{lem:ume'} For every $t,u \in\PT$, the following hold: 
    \begin{enumerate}
        \item $t \le_{\fin} t^{\mathfrak{m}'}$; 
        \item if $t \le_{\fin} u$, then $u \le_{\fin} t^{\mathfrak{m}'}$. 
    \end{enumerate}
    In particular, $t^{\mathfrak{m}'}$ is the unique maximal element in $t^{\uparrow'}:=\{u : t \le_{\fin} u\}$. 
\end{proposition}

\begin{lemma}
    \label{lem:trans}
    For every $t, u,v \in \PT$, if $t \le_{\fin} v $ and $u \le_{\fin}v$, then there is $w \in \PT$ such that $w \le_{\fin} t$ and $w \le_{\fin} u$. 
\end{lemma}

In analogy with the minimum group congruence, we define a relation ${\sigma'_{\PT}}$ by letting $t \, \sigma'_{\PT} \, u$ if there is $w \in\PT$ such that $w \le_{\fin} t$ and $w \le_{\fin} u$. 
\begin{lemma}
    \label{lem:primecong}
    The relation ${\sigma'_{\PT}}$ is a congruence on $PT$ such that $\sigma'_{\PT} \subseteq \sigma_{\PT}$. 
\end{lemma}

\begin{proof}
    Reflexivity and symmetry are immediate and transitivity follows from Lemma \ref{lem:trans}. 
    Moreover, $\sigma'_{\PT}$ is a congruence because $\leq_{\fin}$ is compatible.
\end{proof}

\begin{example}
    \label{ex:fininf}
    The infinite permutation tree $e=\{1^n : n \in \mathbb{N}\}$ 
    is idempotent. 
    Now, $e^{\mathfrak{m}}=\bar \epsilon$ while $e^{\mathfrak{m}'}=e$. 
    This implies that $e \nleq_{\fin} \bar \epsilon$
    and that $\sigma'_{\PT}$ is strictly contained in $\sigma_{\PT}$.
\end{example}

\begin{remark}
    \label{rem:fininf}
    We observe that if $t \in \PT$ is infinite, then $t^{\mathfrak{m}'}$ is infinite, while $t^{\mathfrak{m}}$ may be finite. 
    As a consequence, for an idempotent $e \in E$, $e \, \sigma'_{\PT} \, \bar \epsilon$ iff $e \in E_{\fin}$. 
\end{remark}

The congruences ${\sigma'_{\PT}}$ and ${\sigma_{\PT}}$ coincide on finite permutation trees. 
In particular, $\FPT /{\sigma_{\PT}} \simeq \FPT /{\sigma'_{\PT}}$. 
We characterise the congruence $\sigma'_{\PT}$ in a way analogous to $\sigma_{\PT}$. 

\begin{proposition}
    \label{prop:maxprime}
    For all $t,u \in\PT$, we have: $t \, \sigma'_{\PT} \, u$ iff $t^{\mathfrak{m}'} = u^{\mathfrak{m}'}$. 
\end{proposition}
\begin{proof}
$(\Rightarrow)$
    If $t \, \sigma'_{\PT} \, u$ then, by definition, there is $w \in\PT$ such that $w \le_{\fin} t$ and $w \le_{\fin} u$. 
    Therefore $w^{\mathfrak{m}'}=t^{\mathfrak{m}'}=u^{\mathfrak{m}'}$. 
$(\Leftarrow)$
    If $t^{\mathfrak{m}'} = u^{\mathfrak{m}'}$ holds, then $t \le_{\fin} t^{\mathfrak{m}'}$ and $u\le_{\fin}u^{\mathfrak{m}'}=t^{\mathfrak{m}'}$. 
    By Lemma \ref{lem:trans}, there exists $w$ such that $w \le_{\fin} t$ and $w \le_{\fin}u$, whence $t \, \sigma'_{\PT} \, u$. 
\end{proof}

By Remark \ref{rem:fininf}, the group of units of $\PT/\sigma'_{\PT}$ is the group $\FPT/\sigma_{\PT}$. 

%% file: ptlambda.tex
In this section we apply and adapt the previous results to the study of inverse monoids and groups in the $\lambda$-calculus.
The natural order ${\le}$ and the order ${\le_{\fin}}$ on permutation trees are closely related to notions of $\eta$-expansion.
The general construction extracting a canonical group from an inverse monoid recovers the invertible $\lambda$-terms in various $\lambda$-theories.

By Proposition \ref{prop:embed}, the image $\mathscr{H\!\!P}$ of the embedding $(-)^+: \PT \to \mathscr{B}$
is an $F$-inverse monoid isomorphic to the monoid $\PT$ of permutation trees. 
Thus, we will transfer all the notions introduced for $\PT$ to $\mathscr{H\!\!P}$, including the natural order, the maximal elements, and the minimum group congruence. 

Restricting the embedding $(-)^+$ to recursively enumerable permutation trees yields an $E$-unitary inverse monoid $\mathscr{H\!\!P}_{\mathrm{rec}}$, where $\mathscr{H\!\!P}_{\mathrm{rec}} = \{\BT(M) : M \in \HP\}$. 
This monoid is not $F$-inverse, since $\BT(M)^\mathfrak{m}$ is not in general recursively enumerable.
We note that $\mathscr{H\!\!P}_{\mathrm{rec}} \simeq \HP/\th B$. 
The set of idempotents of $\HP/\th B$ is the set $\mathcal{I}^{\eta}_{\omega}/\th B$ of possibly infinite $\eta$-expansions of $\mathbf I$ modulo $\th B$
and it forms a $\land$-semilattice (see also \cite[Thm.~11.22]{BM22}).

Restricting $(-)^+$ to finite permutation trees yields an inverse monoid, $\FHP/\boldsymbol{\lambda}$, since $\th B$ and $\boldsymbol{\lambda}$ coincide on $\beta$-normalisable $\lambda$-terms. 
We recover the fact that the idempotent elements $\mathcal{I}^{\eta}/\boldsymbol{\lambda}$ form a $\land$-semilattice \cite{IN03}.

\begin{proposition}
    \label{thm:zeta}
     There exist minimal $\lambda$-theories $\th T, \th T'$ such that $\HP / \th T$ and $\FHP / \th T'$ are groups. 
     Moreover, $\th T'=\boldsymbol{\lambda \eta} \subsetneq \th T \subseteq\th H^*$. 
\end{proposition}
\begin{proof}
    By Lemma \ref{lem:mgc} the $\lambda$-theory $\th T$ is  axiomatised by $\{M=\mathbf I: M\in \mathcal I^\eta_\omega \}$ and the $\lambda$-theory $\th T'$  by $\{M=\mathbf I: M\in \mathcal I^\eta\}$. 
    We have $\th T'=\boldsymbol{\lambda\eta}$, because $\th T'\vdash \mathbf{1}= \mathbf I$ and $\boldsymbol{\lambda\eta} \vdash M=\mathbf I$ for every $M\in \mathcal I^\eta$. 
    Moreover, $\th T'$ is strictly contained into $\th T$, because every $M\in \mathcal I^\eta_\omega\setminus \mathcal I^\eta$ has no $\beta\eta$-normal form, and therefore $M=\mathbf I$ cannot be proven in $\th T'$. 
\end{proof}

We conjecture that the theory $\th T$ of the previous proposition is strictly contained in $\th H^*$. 

Proposition \ref{prop:key} is central to determining invertible elements modulo $\lambda$-theories.
In a $F$-inverse monoid $S$, the minimum group congruence $\sigma_S$ collapses all idempotents in the $\sigma_S$-class of the unit $1$, so every $t \in S$ becomes invertible modulo $\sigma_S$ with inverse $t^*$.
Since maximal elements are canonical representatives of $\sigma_S$-classes, they form a group under $t \cdot u = (tu)^{\mathfrak m}$, which is isomorphic to $S/\sigma_S$ by Lemma \ref{lem:max}.
It is then natural to ask what $\sigma_S$ corresponds to in the $\lambda$-calculus.

\subsection{Hereditary permutations and the theory $\th H^*$}
Two $\lambda$-terms are equal in $\th H^*$ whenever their Böhm trees coincide up to infinite $\eta$-expansions \cite[Thm.~16.2.7]{B84}.
We denote by ${\le^\eta_\omega}$ the preorder on $\mathscr{B}$ such that 
$T\le^\eta_\omega U$ holds precisely when $U$ is obtained from $T$ by performing countably many possibly infinite $\eta$-expansions \cite[Def.~11.11]{BM22}. 
Moreover, we denote by $\nf_{\eta !}(T)$ the Nakajima tree of a Böhm-like tree $T$ \cite[Def.~2.60]{BM22}.
One calculates $\nf_{\eta !}(T)$ by eliminating from $T$ all finite and infinite $\eta$-expansions.

\begin{theorem}\emph{\cite[Prop.~11.12]{BM22}}
    \label{thm:hstar}
    For every $M, N \in \Lambda$ the following are equivalent: 
    \begin{enumerate}
        \item $\th H^* \vdash M = N$; 
        \item $\nf_{\eta !}(\BT(M)) = \nf_{\eta !}(\BT(N))$;  
        \item there exists $T \in \mathscr{B}$ such that 
        $\BT(M) \le^\eta_\omega T \text{ and } \BT(N) \le^\eta_\omega T$.
    \end{enumerate}
\end{theorem}

\begin{lemma}
    \label{lem:order-eta}
    If $M, N \in \HP$, then we have: 
    \begin{enumerate}
        \item $\BT(M) \geq \BT(N)$ iff $\BT(M) \le^\eta_\omega \BT(N)$; 
        \item $\nf_{\eta !}(\BT(M))=\BT(M)^{\mathfrak{m}}$.
    \end{enumerate} 
\end{lemma}
\begin{proof}
    (1) follows from the characterisation of the natural order, Lemma \ref{lem:order}, and \cite[Def.~11.11]{BM22}. 
    (2) follows from (1), Proposition \ref{prop:embed}, Definition \ref{def:max} and \cite[Def.~2.60]{BM22}.
\end{proof}

Remark that the natural order ${\le}$ in $\mathscr{H\!\!P}$ is the opposite of ${\le^\eta_\omega}$. 

\begin{theorem}
    \label{prop:etainf} 
    For any $M,N \in \HP$, the following are equivalent: 
    \begin{enumerate}
    \item $\th H^* \vdash M = N$; 
    \item $\BT(M) \, \sigma_{\mathscr{H\!\!P}} \, \BT(N)$; 
    \item $\BT(M)^{\mathfrak{m}} = \BT(N)^{\mathfrak{m}}$. 
    \end{enumerate}   
\end{theorem}
\begin{proof}
    $(1) \iff (3)$ by Theorem \ref{thm:hstar} and Lemma \ref{lem:order-eta}(2). 
    $(2) \iff (3)$ by Lemma \ref{lem:max}. 
\end{proof}

The preceding analysis gives a characterisation of the invertible elements in $\th H^*$.
Recall that the monoid $\HP/\th B$ is isomorphic to $\mathscr{H\!\!P}_{\mathrm{rec}}$. 

\begin{theorem} 
    \label{thm:invhstar}
    Let $S:=\mathscr{H\!\!P}_{\mathrm{rec}}$. 
    The following groups are isomorphic:
\begin{enumerate}
\item the group $\HP/\th H^*$ of invertible elements of $\th H^*$;
\item the group $S/\sigma_S$;
\item the group of maximal elements of $S$ in the $F$-inverse monoid $\mathscr{H\!\!P}$. 
\end{enumerate}
\end{theorem}

\begin{proof}
    By Theorem \ref{prop:etainf} there is a bijective correspondence between the equivalence classes of $\sigma_S$ and the equivalence classes of $\HP$ in $\th H^*$. 
    By Lemma \ref{lem:max} there is a bijective correspondence between maximal elements of $\mathscr{H\!\!P}$ and equivalence classes of $\sigma_{\mathscr{H\!\!P}}$. 
    The maximal elements of $\mathscr{H\!\!P}_{\mathrm{rec}}$ in $\mathscr{H\!\!P}$ form a subgroup of the group of maximal elements of $\mathscr{H\!\!P}$.
\end{proof}

\begin{theorem}
    \label{thm:hpovert}
    Let $\th T$ be a $\lambda$-theory with $\th B \subseteq \th T$. 
    Then $\HP/\th T$ is an $E$-unitary inverse submonoid of the semigroup $\Lambda / \th T$ with composition. 
\end{theorem}
\begin{proof}
    Since $\th H \subseteq \th B \subseteq \th T$, $\th T$ is sensible and thus $\th T \subseteq \th H^*$. 
Therefore, we have the following chain of surjective homomorphisms of monoids: 
$\HP/ \th B \xrightarrow{f} \HP /\th T \xrightarrow{g} \HP /\th H^*$. 
As $\HP / \th B$ is an inverse monoid, $\HP / \th T$ is such by Remark \ref{rem:surj-homo}.  
We now prove that $\HP /\th T$ is $E$-unitary. 
If $e \le t$ in $\HP /\th T$ with $e$ idempotent, then $g(e) \le g(t)$ in $\HP/\th H^*$ with $g(e)$ idempotent. 
Since $\HP/\th H^*$ is a group, $g(e) = g(t)$ and $g(e)=\mathbf I$, so that $g(t)=\mathbf I$. 
Let $t' \in \HP/\th B$ be such that $f(t')=t$. 
Then $(g \circ f)(t')=\mathbf I$. 
By Theorem \ref{thm:invhstar} $\ker (g \circ f) = \sigma_{\HP/\th B}$, so that $t' \, \sigma_{\HP/\th B} \, \mathbf I$. 
By Lemma \ref{lem:mgc}, this implies that $t'$ is idempotent in $\HP/\th B$. 
Consequently, $t$ is idempotent in $\HP /\th T$ as desired. 
\end{proof}

\subsection{Finite hereditary permutations and the theory $\boldsymbol{\lambda \eta}$}
Finite permutation trees are in bijection with $\beta$-normal forms of finite hereditary permutations. 
We show that the natural order on the finite permutation trees perfectly captures $\eta$-reduction.

\begin{proposition} 
    \label{prop:eta}
    Let $t,u \in \FPT$. Then: 
    \begin{enumerate}
        \item $t\prec u$ iff $t^+ \rightarrow_\eta u^+$; 
        \item $t \le u$ iff $t^+ \twoheadrightarrow_{\eta} u^+$. 
    \end{enumerate}
\end{proposition}

\begin{proof}
 (1) follows from Proposition \ref{prop:embed} and Lemma \ref{lem:singleton2}; (2) from (1) and Lemma \ref{lem:leqprec}.
\end{proof}

By the above proposition, when $t \in\FPT$, the \lam-term $(t^{\mathfrak{m}})^+$ is the $\beta\eta$-normal form of $t^+$.  
\begin{proposition}
    \label{prop:etasigma}
    Let $t,u \in \FPT$. The following are equivalent: 
        \begin{enumerate}
    \item $\boldsymbol{\lambda \eta} \vdash t^+ = u^+$; 
    \item $ t \, \sigma_{\FPT} \, u$; 
    \item $t^{\mathfrak{m}} = u^{\mathfrak{m}}$. 
    \end{enumerate}   
\end{proposition}
\begin{proof}
    The equivalence between (1) and (2) follows from Proposition \ref{prop:eta} and the definition of $\sigma_{\FPT}$; 
    the equivalence between (2) and (3) follows from Lemma \ref{lem:max}. 
\end{proof}

\begin{theorem} 
    \label{thm:lambdaeta}
    Let $S$ be $\FHP /\boldsymbol{\lambda}$.
    The following groups are isomorphic:
\begin{enumerate}
\item the group $\FHP/\boldsymbol{\lambda\eta}$ of invertible elements of $\boldsymbol{\lambda\eta}$; 
\item the group $S/\sigma_S$; 
\item the group of maximal elements of $S$.
\end{enumerate}
\end{theorem}

\begin{proof}
 By Proposition \ref{prop:etasigma}, $\FHP/\boldsymbol{\lambda\eta} \simeq S/\sigma_S$ and by Lemma \ref{lem:max} the group $S/\sigma_S$ is isomorphic to the the group of maximal elements of $S$. 
\end{proof}

\begin{theorem}
    \label{thm:fhpovert}
    For every $\lambda$-theory $\th T$, $\FHP/\th T$ is an $F$-inverse submonoid of $(\Lambda / \th T, \circ)$. 
\end{theorem}
\begin{proof}
Since $\FHP$s are $\beta$-normal forms, and $\th T$ cannot identifty $\beta$-normal forms that have distinct $\beta \eta$-normal forms, 
we have the following chain of surjective homomorphisms of monoids: 
$\FHP/ \boldsymbol{\lambda} \xrightarrow{f} \FHP /\th T \xrightarrow{g} \FHP /\boldsymbol{\lambda \eta}$. 
As $\FHP/ \boldsymbol{\lambda}$ is an inverse monoid, $\FHP / \th T$ is such. 
We prove that $\FHP /\th T$ is $F$-inverse. 
    Let $t \in \FHP/\boldsymbol{\lambda}$. 
    We prove that $f(t^\mathfrak{m})$ is maximal. 
    If, towards a contradiction, there is $u \in \FHP/\boldsymbol{\lambda}$ such that $f(t^{\mathfrak{m}}) < f(u)$ in $\FHP /\th T$, 
    then $(g \circ f)(t^{\mathfrak{m}}) \le (g \circ f)(u)$ in $\FHP/\boldsymbol{\lambda \eta}$, so that, as $\FHP/\boldsymbol{\lambda \eta}$ is a group, $(g \circ f)(t^{\mathfrak{m}}) = (g \circ f)(u)$. 
    By Theorem \ref{thm:lambdaeta} $\ker (g \circ f) = \sigma_{\FHP/\boldsymbol{\lambda}}$, so that $t^{\mathfrak{m}} \, \sigma_{\FHP/\boldsymbol{\lambda}} \, u$. 
    By definition of $\sigma_{\FHP/\boldsymbol{\lambda}}$, there is $w \in \FHP/\boldsymbol{\lambda}$ such that $w \le u, t^{\mathfrak{m}}$.
    This implies that $u \le t^{\mathfrak{m}}$ in ${\FHP/\boldsymbol{\lambda}}$, so that $f(u) \le f(t^{\mathfrak{m}})$ in ${\FHP/\th T}$. 
    Contradiction. 
\end{proof}

\subsection{Invertible terms and the kite}
\begin{figure}
  \begin{tikzpicture}[scale=0.75]
    \node at (-10,0) {~};
    \node (u) at (0,4) {$\boldsymbol{\lambda}$};
    \node (ul) at (-1,3) {$\boldsymbol{\lambda \eta}$};
    \node (l) at (-2,2) {$\boldsymbol{\lambda \omega}$};
    \node (ur) at (1,3) {$\th H$};
    \node (c) at (0,2) {$\th H \boldsymbol{\eta}$};
    \node (dl) at (-1,1) {$\th H \boldsymbol{\omega}$};
    \node (r) at (2,2) {$\th B$};
    \node (dr) at (1,1) {$\th B \boldsymbol{\eta}$}; 
    \node (d) at (0,0) {$\th B \boldsymbol{\omega}$};
    \node (dd) at (0,-1.3) {$\th H^*$};

    \draw (u) -- (ul) -- (l) -- (dl) -- (d) -- (dd)
          (u) -- (ur) -- (r) -- (dr) -- (d)
          (ul) -- (c) -- (dl)
          (ur) -- (c) -- (dr);
  \end{tikzpicture}
  \caption{The kite}
  \label{fig:kite}
\end{figure}

A long-standing question concerned the inclusion relationship between $\th H^+$ and $\th B {\boldsymbol{\omega}}$.
This question has now been settled: the inclusion $\th B {\boldsymbol{\omega}} \subseteq \th H^+$ was proved in \cite{BMPR}, and the equality $\th B {\boldsymbol{\omega}} = \th H^+$ was subsequently established in \cite{IMP17}.
Figure~\ref{fig:kite} depicts a subset of $\lambda$-theories ordered by reverse inclusion, commonly referred to as Barendregt's `kite'.
In this section we complete Theorem \ref{thm:inv} by characterising invertible terms for all the $\lambda$-theories of the kite.  

Similarly to $\th H^*$, equality in $\th H^+$ can also be characterised via a preorder on the set of B\"ohm-like trees.
We denote by ${\le^\eta}$ the preorder on $\mathscr{B}$ such that 
$T\le^\eta U$ holds precisely when $U$ is obtained from $T$ by performing countably many finite $\eta$-expansions \cite[Def.~11.8]{BM22}. 
Moreover, we denote by $\nf_{\eta}(T)$ the 
$\eta$-normal form of a Böhm-like tree $T$ (see \cite[Def.~4]{BB02} and \cite[Def. 2.56]{BM22}). 
Then, for every $W \in \mathscr{B}$, if $W \le^\eta T$, then $\nf_{\eta}(T) \le^\eta W$. 

The next result was proved in \cite{CDZ}. 
See also \cite[Prop.~11.2.20]{DRP} and \cite[Prop.~11.9]{BM22}. 

\begin{theorem}
    \label{thm:hplus}
    For every $M, N \in \Lambda$ the following are equivalent: 
    \begin{enumerate}
        \item $\th H^+ \vdash M = N$; 
        \item $\nf_{\eta}(\BT(M)) = \nf_{\eta}(\BT(N))$;  
        \item there exists $T \in \mathscr{B}$ such that 
        $\BT(M) \le^\eta T$ and $\BT(N) \le^\eta T$.
    \end{enumerate}
\end{theorem}

We now establish when two hereditary permutations are equal in $\th H^+$.  

\begin{lemma}
    \label{lem:orderprime-eta}
    If $M, N \in \HP$, then we have: 
    \begin{enumerate}
        \item $\BT(N) \le_{\fin} \BT(M)$ iff $\BT(M) \le^\eta \BT(N)$; 
        \item $\nf_{\eta}(\BT(M))=\BT(M)^{\mathfrak{m}'}$.
    \end{enumerate} 
\end{lemma}
\begin{proof}
    (1) is immediate from the definition of ${\le_{\fin}}$ and \cite[Def.~11.8]{BM22}. 
    (2) follows from (1), Proposition \ref{prop:embed}, Definition \ref{def:maxprime} and \cite[Def.~4]{BB02}.
\end{proof}

Remark that the order ${\le_{\fin}}$ in $\mathscr{H\!\!P}$ is the opposite of ${\le^\eta}$. 

\begin{theorem}
    \label{thm:etainfplus} 
    For any $M,N \in \HP$, the following are equivalent: 
    \begin{enumerate}
    \item $\th H^+ \vdash M = N$; 
    \item $\BT(M) \, \sigma'_{\mathscr{H\!\!P}} \, \BT(N)$; 
    \item $\BT(M)^{\mathfrak{m}'} = \BT(N)^{\mathfrak{m}'}$. 
    \end{enumerate}   
\end{theorem}
\begin{proof}
    $(1) \iff (3)$ by Theorem \ref{thm:hplus} and Lemma \ref{lem:orderprime-eta}(2). $(2) \iff (3)$ by Proposition \ref{prop:maxprime}. 
\end{proof}

We now characterise the invertible elements in $\th H^+$. 

\begin{theorem} 
    \label{thm:invhplus}
    Let $M \in \Lambda$. Then $M$ is $\th H^+$-invertible iff $M \in \FHP$. 
\end{theorem}
\begin{proof}
    If $M \in \FHP$, then $M$ is $\th H^+$-invertible, since $\th H^+$ is extensional.  
    Conversely, let $M$ be $\th H^+$-invertible, i.e., there is $N$ such that $\th H^+ \vdash M \circ N = N \circ M = \mathbf I$.
    As $\th H^+ \subseteq \th H^*$, this implies that $\th H^* \vdash M \circ N = N \circ M = \mathbf I$. 
    By Theorem \ref{thm:inv} $M,N \in \HP$. 
    Let $P:=M \circ N$ and $Q:=N \circ M$.
    Let $S$ be the inverse monoid $\mathscr{H\!\!P}_{\mathrm{rec}}$.
    By Theorem \ref{thm:invhstar}, $\BT(P) \, \sigma_S \, \BT(\mathbf I)$ and $\BT(Q) \, \sigma_S \, \BT(\mathbf I)$. 
    By Lemma \ref{lem:mgc}, this implies that $\BT(P)$ and $\BT(Q)$ are idempotent in $\mathscr{H\!\!P}_{\mathrm{rec}}$, i.e., $P$ and $Q$ are possibly infinite $\eta$-expansions of $\mathbf I$. 
    Since $\th H^+ \vdash \mathbf I = P$ and $\th H^+ \vdash \mathbf I = Q$, then, by definition of $\th H^+$, for every context $C[\,]$, 
    $C[P]$ has a $\beta$-nf iff  $C[Q]$ has a $\beta$-nf iff $C[\mathbf I]$ has a $\beta$-nf.
    If we choose the empty context, we derive that $P, Q$ are finite $\eta$-expansions of $\mathbf I$. 
    Therefore $M,N \in \FHP$. 
\end{proof}

\begin{lemma}
    \label{lem:giulio}
    Let $\th T$ be non-extensional. If $\th T \vdash D = \mathbf{I}$, for some $D \in \mathcal{I}^\eta_\omega$, 
    then $D =_{\beta} \mathbf I$.  
\end{lemma}

\begin{proof}
    For the sake of contradiction, if $D \neq_{\beta} \mathbf{I}$, then $D=_{\beta}\lambda uv.D'$ for some $D'$. 
    Then $\th T \vdash \mathbf{1} = \lambda xy. \mathbf I xy = \lambda x y. Dxy = \lambda x y. (\lambda uv.D') xy = D = \mathbf I$, so that $\th T$ is extensional.  
\end{proof}

\begin{proposition} Let $\th T$ be a \lam-theory. 
    \begin{enumerate}
        \item If $\th T$ is semisensible and $\th T$ is not extensional, then $M \in \Lambda$ is $\th T$-invertible iff $M=_{\beta}\mathbf I$.
        \item If $\boldsymbol{\lambda \eta} \subseteq \th T \subseteq \th H^{+}$, then $M \in \Lambda$ is $\th T$-invertible iff $M \in \FHP$.
    \end{enumerate}
\end{proposition}
\begin{proof}
    (1) Let $M$ be $\th T$-invertible with inverse $N$. 
    Since $\th T \subseteq \th H^*$, we have $M, N \in \HP$ by Theorem \ref{thm:inv}, and  
    $ \th H^* \vdash M \circ N = N \circ M = \mathbf I$. 
    Then, following the proof of Theorem \ref{thm:invhplus}, we
    can prove that $M \circ N \in \mathcal{I}^\eta_\omega$ and $N \circ M \in \mathcal{I}^\eta_\omega$.
    By Lemma \ref{lem:giulio} applied to $\th T$, we get that $M \circ N =_{\beta} N \circ M=_{\beta} \mathbf I$. 
    As the only invertible term modulo $\boldsymbol{\lambda}$ is $\mathbf I$, $M =_\beta N =_\beta \mathbf I$. 
    (2) follows from Theorems \ref{thm:invhplus} and \ref{thm:inv}. 
\end{proof}
Thus we confirm Barendregt's conjecture \cite[p. 547]{B84} stating that $M$ is $\th T$-invertible, for $\th T=\th B \boldsymbol{\eta}$ and $\th T= \th B \boldsymbol{\omega}$, exactly when $M \in \FHP$.

%% file: conclusion.tex
Inverse semigroups were first employed in the $\lambda$-calculus to study properties of $\beta$-reduction \cite{DR93}.
Subsequently, they played a role in Girard's Geometry of Interaction \cite{G87}.
More recently, Goubault-Larrecq, inspired by the work \cite{DR93,DR99}, has shown that a category constructed from an inverse semigroup models a fragment of linear logic \cite{JGL}.
Hofmann and Mislove employed the theory of inverse semigroups to prove that every compact Hausdorff topological model of $\lambda$-calculus is degenerate \cite{HM}.
We have shown that the theory of inverse semigroups improves our understanding of the interaction between invertible and idempotent elements, which is essential for characterising \lam-terms invertible modulo \lam-theories. 
Looking ahead, we aim to characterise elements that are invertible modulo \lam-theories lying strictly between $\th H^+$ and $\th H^*$. 
The role played by $E$ and $E_{\fin}$ in $\th H^*$ and $\th H^+$, respectively, is taken on by other well-behaved semilattices of idempotents of $\PT$. 
We conjecture that there exist \lam-theories between $\th H^+$ and $\th H^*$ whose group of invertible elements is strictly between $\FHP$ and $\HP$.

%% file: appendix.tex
\appendix 
\section{Proofs}
We collect here all the missing proofs.

\subparagraph*{Proof of Lemma \ref{lem:norm}}
    Let $t=\<\pi; t_1, \ldots, t_n\>$ and $u=\<\rho; u_1, \ldots, u_m\>$. 
    \begin{enumerate}
    \item Assume that $t < u$; then either $n>m$ or $n=m$. 
    In the first case, 
    \begin{equation*}
        \norm t =1+\sum_{j=1}^{m} \norm {t_j} + \sum_{j=m+1}^{n} \norm {t_j} > 1+\sum_{j=1}^{m} \norm {t_j} \geq 1+\sum_{j=1}^{m} \norm {u_j}. 
    \end{equation*}
    In the second case, there exists $i$ such that $t_i < u_i$. 
    Then 
    $$\norm t=1+\sum_{j=1}^{n} \norm {t_j} > 1+\sum_{j=1}^{m} \norm {u_j}\text{,}$$ 
    because $\norm {t_i}> \norm {u_i}$ and $\norm {t_j}\geq \norm {u_j}$ for every $j\neq i$.
    \item The proof is by induction on $t$. 
    Observe that 
    \begin{equation*}
        \norm {t^*} = 1+ \sum_{i=1}^{n} \norm{(t^*)_i} 
        = 1+ \sum_{i=1}^{n} \norm{(t_{\pi^{-1}i})^*} 
        = 1+ \sum_{i=1}^{n} \norm {t_{\pi^{-1}i}}
        = 1+ \sum_{i=1}^{n} \norm{t_i}
    \end{equation*}
    which equals $\norm  t$, because $\pi^{-1}$ is a permutation of $n$.
    \item 
    We show (3) by induction on $\norm t +\norm u$. 
If $\norm t +\norm u=2$, then $t =u=\bar\epsilon$, and the conclusion is trivial.
If $\norm t +\norm u>2$, then we have two cases: (a) $n\geq m$ and (b) $m\geq n$.
\begin{enumerate}
    \item If $m=0$, then $u=\bar\epsilon$ and $tu=t$. 
    Then the conclusion is trivial. 
    Assume now $m>0$.
Since $(tu)_i = u_i t_{\rho i}$ for every $1\leq i\leq m$ and $(tu)_i=t_i$ for every $\rd(u)< i\leq \rd(t)$, then by induction hypothesis $\norm{u_i}, \norm {t_{\rho i}}\leq \norm{(tu)_i}$ for every $1\leq i\leq m$. 
Then we have:
\begin{equation*}
    \norm{tu}= 1+ \sum_{i=1}^{n} \norm{tu_i} = 1+ \sum_{i=1}^{m} \norm{tu_i} + \sum_{i=m+1}^{n} \norm{tu_i} \geq 1+ \sum_{i=1}^{m} \norm{u_i}=\norm{u}
\end{equation*}
and 
\begin{equation*}
    \norm{tu}=  1+ \sum_{i=1}^{m} \norm{(tu)_i} + \sum_{i=m+1}^{n} \norm{t_i}\geq 1+ \sum_{i=1}^{m} \norm{u_{\rho i}}+ \sum_{i=m+1}^{n} \norm{t_i}=\norm t
\end{equation*}
because $\rho$ is a permutation of $m$ and $\sum_{i=1}^{m} \norm{t_{\rho i}}=\sum_{j=1}^{m} \norm{t_j}$.
Finally 
\begin{align*}
    \norm{tu} & =1+ \sum_{i=1}^{m} \norm{(tu)_i} + \sum_{i=m+1}^{n} \norm{t_i}  \leq 2+ \sum_{i=1}^{m} (\norm{u_i}+\norm{t_{\rho i}})+ \sum_{i=m+1}^{n} \norm{t_i} \\
    & = \left (1+ \sum_{i=1}^{m} \norm{u_i} \right )+ \left (1+\sum_{i=1}^{m} \norm{t_{\rho i}}  +\sum_{i=m+1}^{n} \norm t_i\right ) \\
    & = \norm{u}+\left (1+ \sum_{j=1}^{m}\norm{t_j} +\sum_{i=m+1}^{n} \norm{t_i}\right )= \norm{u}+\norm{t}\text{,}
\end{align*}
because $\rho$ is a permutation of $m$.
\item Assume $m>0$. 
Since $(tu)_i = u_i t_{\rho i}$ for every $1\leq i\leq m$ and $t_{\rho i}=\bar\epsilon$ for every $\rho i>n$, then $(tu)_i=u_i$ for every $i$ such that $\rho i>n$.
Let $A=\{i: \rho i>n, 1\leq i\leq m\}$ and $B=\{i: \rho i\leq m, 1\leq i\leq n \}$. 
Then we have: 
    \begin{equation*}
        \norm{tu}=  1+ \sum_{i\in A} \norm{(tu)_i} + \sum_{i\in B} \norm{(tu)_i}= 1+ \sum_{i\in A} \norm{u_i} + \sum_{i\in B} \norm{(tu)_i}\text{,}
    \end{equation*}
so that 
\begin{equation*}
    \norm{tu}\geq 1+ \sum_{i\in A} \norm{u_i} + \sum_{i\in B} \norm{u_i}=\norm u
\end{equation*}
and 
\begin{equation*}
    \norm{tu}\geq 1 + \sum_{i\in B} \norm{(tu)_i} =1+\sum_{i\in B} \norm{t_{\rho i}}=\norm t \text{,}
\end{equation*}
because $\rho$ is a permutation of $m$ and $\sum_{i\in B} \norm{t_{\rho i}}=\sum_{j=1}^{n} \norm{t_j}$.
Finally, 
\begin{equation*}
    \norm{tu}\leq 2+ \sum_{i\in A} \norm{u_i} + \sum_{i\in B} \norm{u_i}+ \sum_{i\in B}  \norm{t_{\rho i}}=\norm u+ 1+ \sum_{i\in B}  \norm{t_{\rho i}}= \norm u+\norm t
\end{equation*}
because $\sum_{i\in B} \norm{t_{\rho i}}=\sum_{j=1}^{n} \norm{t_j}$.
\end{enumerate}
\item The first equality is proved by induction on $t$. 
Since $(tt^*)_i=  (t_{\pi^{-1}i})^* t_{\pi^{-1}i}$ and $(t^* t)_i=t_i(t_i)^*$, by induction hypothesis $\norm{\sk(t_i)} = \norm{t_i t_i^*}$. 
The second follows from (2). \qed
\end{enumerate}

\subparagraph*{Proof of Lemma \ref{lem:singleton2}}
    First of all observe that, since $t \prec u$ implies $t \le u$, then $\pi = \rho$. 
 If $u=\bar\epsilon$, then $t$ is idempotent. 
 If $t$ has more than two nodes, it can be easily seen that there is $z \in \FPT$ such that $t < z < \bar \epsilon$. 
 Then $t$ has two nodes and (ii) holds. 
 If $u\neq \bar\epsilon$, then we have two cases: (a) $n = m$ and (b) $n > m$. 
 \begin{alphaenumerate}
    \item If $n = m$, then $t_i \leq  u_i$ for every $1\leq i\leq n$ and it must exist $k$ such that $t_k <  u_k$. 
    If there is $w\neq t_k,u_k$ such that $t_k <  w< u_k$, then $t < \<\pi; t_1, \ldots, t_{k-1},w,t_{k+1}, \ldots, t_n\> <u$, 
    leading to a contradiction. 
    Therefore, $t_k\prec u_k$.
    If there is another $j\neq k$ such that $t_j\prec u_j$, then
    $t < \<\pi; t_1,\ldots,t_j,\ldots,t_{k-1},u_k,t_{k+1},\ldots,t_n\> <  u$, leading again to a contradiction.
    In conclusion, (i) holds.
    \item If $n >m$, similarly, (ii) has to hold, otherwise by Lemma \ref{lem:order} there would be $z \in \FPT$ such that $t < z <u$. 
 \end{alphaenumerate}
We prove the converse. 
Assume that (i) holds and that $t \le  w\le  u$ for some $w=\<\pi; w_1, \ldots,w_p\> \in \FPT$.  
We immediately get that either $w=t$ or $w=u$.
Assume that (ii) holds and that $t \le  w \le   u$. 
We have two cases: $n= p$ or $n= p+1$. 
In the first case we have $t=w$, because $t_i\leq  w_i\leq  u_i=t_i$ for every $1\leq i\leq m$ and $t_n=\bar\epsilon\leq  w_n$. 
In the second case $w=u$, because $t_i=u_i\leq  w_i\leq  u_i $ for every $1\leq i\leq m=p$. 
\qed

\subparagraph*{Proof of Lemma \ref{lem:oprime}}
Let $t:=\<\pi; t_1, \ldots, t_n\>$ and $u:=\<\rho;u_1, \ldots, u_m\>$.
\begin{enumerate}
    \item   
    We prove that there is $f \in E_{\fin}$ such that $uf =t$ by induction on $e \in E_{\fin}$.
    If $e=\bar \epsilon$, then $f = \bar \epsilon$. 
    Let $e=\<\iota; e_1, \ldots, e_k\>$ with $1 \le k \le n$.
    Assume that $m \ge k$; the other case is similar. 
    If $t=eu$ then $\pi = \rho$ and $t_i = u_i e_{\rho i}$ for $1 \le i \le k$ and $t_i = u_i$ for $k < i \le n$. 
    As $e_{\rho i} \in E_{\fin}$ for every $1 \le i \le k$, there are $f_1, \ldots, f_k \in E_{\fin}$ such that $t_i = f_i u_i$.  
    Let $f:=\<\iota;f_1, \ldots, f_k\>$; then $uf  = \<\pi; f_1 u_1 , \ldots, f_k u_k , u_{k+1}, \ldots, u_m\> =t$. 
    \item  
    The proof is by coinduction on $\le_{\fin}$. 
    Let $t=ue$ for some $e \in E_{\fin}$. 
    Then $\pi = \rho$ and $n=\max\{m,\rd(e)\}\geq m$. 
    Moreover, as $t_i = e_i u_i$, for $1\leq i\leq m$, then by (\ref{lem:order-prime}), there are $f_1, \ldots, f_m \in E_{\fin}$ such that 
    $t_i = u_i f_i$ for $1\leq i\leq m$. 
    As $t_i=e_i$ for $m< i\leq n$ the conclusion follows from the coinductive hypothesis.
    \item The proof is by induction on $f \in E_{\fin}$. 
    If $f=\bar\epsilon$, then $t=\bar\epsilon \in E_{\fin}$.
    Let $f = \<\iota; f_1, \ldots,f_m\>$, $m \ge 1$, such that $f\le_{\fin} t$. 
    Then $t=\<\iota; t_1, \ldots, t_n\>$ with $n \le m$ and $f_i \le_{\fin} t_i$ for $1 \le i \le n$. 
    As $f_i \in E_{\fin}$, $t_i \in E_{\fin}$ for every $1 \le i \le n$, so that $t \in E_{\fin}$. \qed
\end{enumerate}
    
\subparagraph*{Proof of Proposition \ref{lem:ume'}}
    If $t$ is idempotent, (1) is trivial and (2) follows from Lemma \ref{lem:oprime}(\ref{lem:funi}).  
    Otherwise let $t=\<\pi; t_1, \ldots, t_n\>$ and let $k$ be the greatest $i$ such that $t_k \notin E_{\fin}$ or $\pi k \neq k$, and $t_i \in E_{\fin}$, $\pi i = i$ for every $k < i \le n$. 
    Then $t^{\mathfrak{m}'}=\<\pi; (t_1)^{\mathfrak{m}'}, \ldots, (t_k)^{\mathfrak{m}'}\>$. 
    As to (1), by coinduction, $ t_i \le (t_i)^{\mathfrak{m}'}$ for $1 \le i \le k$, and $t_i \in E_{\fin}$ for $k < i \le n$, so that $t \le t^{\mathfrak{m}'}$. 
    Concerning (2), let $u=\<\rho; u_1, \ldots, u_m\>$ such that $t \le u$.
    Without loss of generality, we assume that $n=m$, so that $k \le m$.
    Moreover, we have $\pi = \rho$, so that
    by coinductive hypothesis, $u_i \le (t_i)^{\mathfrak{m}'}$ for $1 \le i \le k$. 
    The fact that $t_i \le u_i$ and $t_i \in E_{\fin}$ for $k< i \le n$, implies by Lemma \ref{lem:oprime}(\ref{lem:funi}) that $u_i \in E_{\fin}$ for $k< i \le m$. 
    We conclude that $u \le t^{\mathfrak{m}'}$. 
\qed 

\subparagraph*{Proof of Lemma \ref{lem:trans}}
    Let $t=\<\pi; t_1, \ldots, t_n\>$, $u=\<\rho; u_1, \ldots, u_m\>$ and $v=\<\tau; v_1, \ldots, v_k\>$. 
    By definition of ${\le_{\fin}}$ we have $\pi=\rho=\tau$ and $k \le n,m$. 
    Without loss of generality we assume $n \le m$. 
    We define $w:=\<\pi; w_1, \ldots, w_m\>$ where $w_i$ is 
    \begin{itemize}
        \item for $1 \le i \le k$ a $\le_{\fin}$-lower bound of $t_i$ and $u_i$ obtained coinductively; 
        \item for $k < i \le n$, $w_i:=t_i u_i$; 
        \item for $n < i \le m$, $w_i:=u_i$. 
    \end{itemize}
    Now we verify that $w \le_{\fin} t$: 
      \begin{itemize}
        \item for $1 \le i \le k$ $w_i \le t_i$ by coinductive hypothesis; 
        \item for $k < i \le n$, $w_i=t_i u_i \le_{\fin} t_i$ by Lemma \ref{lem:oprime}. 
        \item for $n < i \le m$, $w_i=u_i \in E_{\fin}$. 
    \end{itemize}
    Similarly, $w \le_{\fin} u$, proving the claimed result. 
    \qed